\documentclass{llncs}

\usepackage[T1]{fontenc}
\usepackage[utf8]{inputenc}

\usepackage{amsmath}
\usepackage{amssymb}

\usepackage{bm}
\usepackage{csquotes}
\usepackage{enumitem}
\usepackage{mathtools}
\usepackage[normalem]{ulem}
\usepackage{color, colortbl}
\usepackage[ruled, vlined, linesnumbered]{algorithm2e}
\usepackage{multirow}
\usepackage{cite}

\SetKwInput{KwData}{Input}
\SetKwInput{KwResult}{Output}

\setlist[itemize]{parsep=2pt, topsep=2pt}
\newcommand\FF{\mathbb{F}}

\newcommand\PP{\mathbb{P}}

\newcommand\KK{\mathbb{K}}

\newcommand\calA{\mathcal{A}}

\newcommand\calC{\mathcal{C}}
\newcommand\calD{\mathcal{D}}

\newcommand\calM{\mathcal{M}}
\newcommand\calO{\mathcal{O}}
\newcommand\calP{\mathcal{P}}

\newcommand\calV{\mathcal{V}}

\newcommand\bfa{{\bm{a}}}

\newcommand\bfc{{\bm{c}}}

\newcommand\bfe{{\bm{e}}}
\newcommand\bfg{{\bm{g}}}
\newcommand\bfh{{\bm{h}}}

\newcommand\bfk{{\bm{k}}}

\newcommand\bfp{{\bm{p}}}

\newcommand\bfr{{\bm{r}}}

\newcommand\bfu{{\bm{u}}}

\newcommand\bfx{{\bm{x}}}
\newcommand\bfy{{\bm{y}}}
\newcommand\bfz{{\bm{z}}}
\newcommand\bfzero{{\bm{0}}}

\newcommand\bfbeta{{\bm{\beta}}}

\newcommand\bfA{{\bm{A}}}
\newcommand\bfB{{\bm{B}}}
\newcommand\bfC{{\bm{C}}}
\newcommand\bfG{{\bm{G}}}
\newcommand\bfH{{\bm{H}}}

\newcommand\bfM{{\bm{M}}}
\newcommand\bfP{{\bm{P}}}
\newcommand\bfV{{\bm{V}}}

\newcommand\bfS{{\bm{S}}}
\newcommand\bfT{{\bm{T}}}
\newcommand\bfX{{\bm{X}}}
\newcommand\bfY{{\bm{Y}}}

\newcommand\ext{\operatorname{Ext}}
\newcommand\Gab{\operatorname{Gab}}
\newcommand\rowspan{\operatorname{RowSp}}
\newcommand\colspan{\operatorname{ColSp}}
\newcommand\Moore{\operatorname{Moore}}
\newcommand\Gr{\operatorname{Gr}}

\newcommand\KeyGen{\mathsf{KeyGen}}
\newcommand\Encrypt{\mathsf{Encrypt}}
\newcommand\Decrypt{\mathsf{Decrypt}}

\newcommand\kpriv{{\bfk_{\rm priv}}}
\newcommand\kpub{{\bfk_{\rm pub}}}

\newcommand\minrank{\textsc{MinRank}}
\newcommand\gabsd{\textsc{Gab-SD}}
\newcommand\corgab{\textsc{CorGab}}

\DeclareMathOperator{\rk}{rk}

\renewcommand\epsilon{\varepsilon}
\newcommand\mydef{\coloneqq}
\newcommand\qbin[3]{
  \left[\begin{smallmatrix} #1 \\[0.3em] #2 \end{smallmatrix}\right]_{#3}
}

\title{RAMESSES,  a Rank Metric Encryption Scheme with Short Keys}

\author{Julien Lavauzelle\inst{1} \and  Pierre Loidreau\inst{2} \and  Ba-Duc Pham\inst{1}}
\institute{Univ Rennes, CNRS, IRMAR - UMR 6625, F-35000 Rennes, France \\\email{\{julien.lavauzelle, ba-duc.pham\}@univ-rennes1.fr}\and
Univ Rennes, DGA MI, CNRS, IRMAR - UMR 6625, F-35000 Rennes, France \\\email{pierre.loidreau@univ-rennes1.fr } }

\date{\today}

\begin{document}

\maketitle

\begin{abstract}
We present a rank metric code-based encryption scheme with key and ciphertext sizes comparable to that of isogeny-based cryptography for an equivalent security level. The system also benefits from efficient encryption and decryption algorithms, which rely on linear algebra operations over finite fields of moderate sizes. The security only relies on rank metric decoding problems, and does not require to hide the structure of a code. Based on the current knowledge, those problems cannot be efficiently solved by a quantum computer. Finally, the proposed scheme admits a failure probability that can be precisely controlled and made as low as possible.
\end{abstract}


\section{Introduction}

With the growing probability of the existence of a near-future quantum computer, it has become important to propose alternatives to existing public-key encryption schemes and key exchange protocols based on number theory. The recent NIST Post-Quantum Cryptography Standardization process motivates proposals in this sense. Along with lattice-based cryptography, code-based cryptography is the most represented among proposals for encryption schemes or key-encapsulation mechanisms (KEMs).
Code-based submissions generically rely on the hardness of decoding problems, either in the Hamming metric or in the rank metric. Hamming metric decoding problems enjoy a long-standing study and few practical improvements for more than fifty years, which ascertain their security.
On the opposite, rank metric decoding problems have been studied for less than twenty years~\cite{CS96}, and their solving complexity is not yet fully stabilized (see the recent results of \cite{BardetBBGNRT09}).
Nevertheless, they benefit from much shorter keys and seem very attractive for practical implementation, culminating in submissions for the NIST standardization process~\cite{ROLLO17, RQC17}. So as to further reduce the key sizes, designers often use specific structures as quasi-cyclicity (equivalent of Module-LWE for lattices) which could be suspected to introduce additional weaknesses~\cite{Loi14}.

In this paper we aim at designing a new one-way encryption scheme featuring very compact keys, based on rank metric decoding problems. The long-standing idea finds origins in~\cite{faure2006new} which was an extended idea of a proposal in Hamming metric~\cite{AugotFiniasz-PKC-PolyReconstruction_2003}. The original rank metric encryption scheme was broken in \cite{GaboritOT18}, and a recent repair was proposed in~\cite{wachter2018repairing}. However it implies to choose a specific code and a syndrome coming from a structured vector of moderate rank, which we want to avoid here.

Inspired from~\cite{faure2006new}, we design a simple one-way encryption scheme with the following strengths.
\begin{itemize}
  \item The security of the scheme only relies on decoding problems in rank metric (such as {\minrank} and {\gabsd}) and does not require to hide the structure of a code. These decoding problems have been --- and are still being --- scrutinized in active research fields.
  \item Especially as a KEM, our proposal enables very small parameters for a given security target. Key sizes are competitive with isogeny-based proposals such as SIKE~\cite{SIKE17}. 
  \item Even if the decryption algorithm is probabilistic, it is easy to control the failure probability and to make it as small as possible without increasing to much the parameters.
\end{itemize}

A remaining weakness would be that underlying problems have been less investigated than others. However, our goal here is also to emulate research in this field to be able to ascertain the security of the scheme.

In a first section we introduce necessary notation and definitions. Then we describe the encryption scheme and we propose sets of parameters for security levels 1, 3, 5 of the NIST competition. Keys and ciphertext sizes are not larger than few hundreds of bytes. In the next section, we prove the consistency of the encryption scheme and we analyze its security by showing to which problems the security can be reduced, and by giving the complexity of algorithms solving these problems.

\section{Preliminaries}

\subsection{Notation and definitions}
\label{subsec:notation}

Throughout the paper, we set $q = 2^n$ for some integer $n \ge 1$, and we let $\FF_q$ denote the finite field with $q$ elements. The field $\FF_q$ can also be viewed as a vector space of dimension $n$ over $\FF_2$. The map $\theta: \FF_q \to \FF_q$, $x \mapsto x^2$, is $\FF_2$-linear and is called the Frobenius automorphism. Its inverse is the $(n-1)$-fold composition $\theta^{n-1} = \theta \circ \dots \circ \theta$. For convenience, we sometimes write $x^{[i]} \mydef \theta^i(x)$, for $i \in [0,n-1] \mydef \{0, \dots, n-1\}$.

Let $\bfbeta = (\beta_1, \dots, \beta_n) \in \FF_q^n$ be a basis of $\FF_q$ over $\FF_2$. We define the extension map
\[
\begin{array}{rclc}
    \ext_\bfbeta : & \FF_q^n & \to & \FF_2^{n \times n} \\
    ~ & \bfa = (a_1, \dots, a_n) & \mapsto & \bfA = ( \bm{\alpha}_1^\top, \dots, \bm{\alpha}_n^\top )
\end{array}
\]
where, for all $1 \le j \le n$, the vector $\bm{\alpha}_j \in \FF_2^n$ consists of coordinates of $a_j \in \FF_q$ in the basis $\bfbeta$, \emph{i.e.\ } $a_j = \sum_{i=1}^n \beta_iA_{i,j}$. In particular, for every $\bfA \in \FF_2^{n \times n}$, we have $\ext_\bfbeta(\bfbeta \bfA) = \bfA$.

The \emph{rank} of $\bfa \in \FF_q^n$, denoted $\rk(\bfa)$, is defined as the rank over $\FF_2$ of its extension matrix $\bfA = \ext_\bfbeta(\bfa)$. Notice that $\rk(\bfa)$ does not depend on the choice of the basis $\bfbeta$. We also define the row space of $\bfa \in \FF_q^n$ with respect to $\bfbeta$ as
\[
\rowspan_\bfbeta(\bfa) \mydef \{ \bfx \ext_\bfbeta(\bfa), \bfx \in \FF_2^n\} \subseteq \FF_2^n\,.
\]
Similarly, the column space of $\bfa \in \FF_q^n$ is $\colspan_\bfbeta(\bfa) \mydef \{ sum_{i=1}^n x_i a_i \mid \bfx \in \FF_2^n\} \subseteq \FF_q$.

We let $\Gr(t, \FF_2^n)$ denote the set of subspaces of $\FF_2^n$ of dimension $t$, which contains $\qbin{n}{t}{2} \mydef \frac{(2^n-1)(2^{n-1}-1) \cdots (2^{n-t+1}-1)}{(2^t-1)(2^{t-1}-1) \cdots (2^1-1)}$ elements. Each subspace $\calV \in \Gr(t, \FF_2^n)$ can be represented by the unique reduced row echelon form (RREF) of any matrix $\bfV \in \FF_2^{n \times n}$ whose row space generates $\calV$. We know from~\cite{SilbersteinE11, Medvedeva12} that this representation can be computed efficiently (in time $\tilde{O}(nt(n-t))$).
Recall that a matrix is in reduced row echelon form if the following holds:
\begin{itemize}[label=--]
\item the index of the pivot (\emph{i.e.}\ the first non-zero coefficient) of row $i$ is strictly larger than the index of the pivot of row $i-1$;
\item all pivots are ones;
\item each pivot is the only non-zero entry in its column.
\end{itemize}
We finally define $\calP_{t,n} \mydef \{ \bfP \in \FF_2^{n \times n} \mid \rk(\bfP) = t, \bfP \text{ is in RREF} \}$.



\subsection{Rank metric codes}

In this paper, we embed $\FF_q^n$ with the \emph{rank metric}: for $\bfa \in \FF_q^n$, the weight of $\bfa$ is defined as $\| \bfa \| \mydef \rk(\bfa)$. We consider $\FF_q$-linear codes, \emph{i.e.}\ $\FF_q$-linear subspaces $\calC \subseteq \FF_q^n$. Notice that the field extension degree $n$ is also the length of the code $\calC$. The dimension of a code $\calC$ is $k = \dim_{\FF_q}(\calC)$, and its minimum (rank) distance is $d = \min_{\bfc \in \calC \setminus \{ \bfzero \}} \| \bfc \|$. A generator matrix (resp. a parity-check matrix) for $\calC$ is a matrix $\bfG \in \FF_q^{k \times n}$ (resp. $\bfH \in \FF_q^{(n-k) \times n}$) such that $\calC = \{ \bfa \bfG, \bfa \in \FF_q^k \}$ (resp. $\bfH \bfc^\top = \bfzero$ for every $\bfc \in \calC$).

Let us define $\bfbeta^{[i]} \mydef (\beta_1^{[i]}, \dots, \beta_n^{[i]}) \in \FF_q^n$. Its Moore matrix is defined as
\[
\Moore_n(\bfbeta) \mydef
\begin{pmatrix}
  ~ & \bfbeta^{[0]} & ~ \\ 
  ~ & \vdots      & ~ \\ 
  ~ & \bfbeta^{[n-1]} & ~
\end{pmatrix} \in \FF_q^{n \times n}
\]
and it is invertible over $\FF_q$. Hence $(\bfbeta^{[0]}, \dots, \bfbeta^{[n-1]})$ is a basis of $\FF_q^n$.

\begin{definition}[Gabidulin code~\cite{Delsarte78, Gabidulin85}]
  Let $\bfg =(g_1, \dots, g_n) \in \FF_q^n$ be an ordered basis of $\FF_q/\FF_2$. The \emph{Gabidulin code} of dimension $k$ with evaluation vector $\bfg$ is the subspace $\Gab_k(\bfg) \subseteq \FF_q^n$ generated by the $k$ first rows of $\Moore_n(\bfg)$.
\end{definition}

Gabidulin codes are optimal codes with respect to the rank metric~\cite{Delsarte78} and they can be efficiently decoded~\cite{Gabidulin85} up to $\lfloor \frac{n-k}{2} \rfloor$ errors. By definition, the submatrix consisting in the $k$ first rows of $\Moore_n(\bfg)$ is a generator matrix for $\Gab_k(\bfg)$. It is also clear that $\Gab_k(\bfg) \subseteq \Gab_{k+1}(\bfg)$ for every $1 \le k \le n$, and we have $\Gab_n(\bfg) = \FF_q^n$. Hence, one can propose the following definition.

\begin{definition}[$\bfg$-degree]
  Let $\bfx \in \FF_q^n$ and $\bfX = \ext_\bfg(\bfx)$. The \emph{$\bfg$-degree of $\bfx$}, denoted $\deg_\bfg(\bfx)$, is the unique integer $\ell \in [0,n-1]$ such that $\bfx \in \Gab_{\ell+1}(\bfg) \setminus \Gab_\ell(\bfg)$. Similarly, one defines the $\bfg$-degree of $\bfX$ as $\deg_\bfg(\bfX) = \deg_\bfg(\bfx)$. 
\end{definition}

In other words, a vector $\bfx \in \FF_q^n$ of $\bfg$-degree $\ell$ can be written
\[
\bfx = \lambda_\ell \bfg^{[\ell]} + \sum_{j=0}^{\ell-1} \lambda_j \bfg^{[j]}
\]
for some non-zero $\lambda_\ell \in \FF_q \setminus \{ 0 \}$ and some $\ell$-tuple $(\lambda_{\ell-1}, \dots, \lambda_0) \in \FF_q^\ell$.

Finally, the dual code $\Gab_k(\bfg)^\perp = \{ \bfa \in \FF_q^n \mid \forall \bfc \in \Gab_k(\bfg),  \sum_{i=1}^n a_i c_i = 0 \}$ is also a Gabidulin code $\Gab_{n-k}(\bfh)$ for some basis $\bfh \in \FF_q^n$ that can be efficiently computed from $\bfg$. In other words, there exists a parity-check matrix for $\Gab_k(\bfg)$ consisting in the $(n-k)$ first rows of a Moore matrix associated to some $\bfh \in \FF_q^n$, see \emph{e.g.}~\cite{Gabidulin85}.

\section{The encryption scheme}
\label{sec:scheme}

\paragraph{System parameters.} Integers $1 \le w, k, \ell, t \le n$ are public parameters and specified according to the desired security level (see Section~\ref{sec:parameters}). We set $q = 2^n$, and we also make public a basis $\bfg$ of $\FF_q/\FF_2$. We let $\bfH$ denote a \emph{fixed} parity-check matrix of $\Gab_k(\bfg)$.

\paragraph{Key generation.} Alice picks uniformly at random a vector $\kpriv \in \FF_q^n$ of rank $w$. As explained in Algorithm~\ref{algo:keygen}, the public key is the syndrome of $\kpriv$ with respect to the parity-check matrix $\bfH$ of $\Gab_k(\bfg)$, and the private key is $\kpriv$.

\begin{algorithm}[t!]
  \KwData{}
  \KwResult{a pair of public/private keys $(\kpub, \kpriv)$}
  Pick $\kpriv \leftarrow_{\$} \{ \bfx \in \FF_q^n, \|\bfx\| = w \}$\\
  Compute $\kpub \in \FF_q^{n-k}$ such that $\kpub^\top = \bfH \kpriv^\top$\\
  Output $(\kpub, \kpriv) \in \FF_q^{n-k} \times \FF_q^n$
  \caption{\label{algo:keygen}$\KeyGen(1^\lambda)$}
\end{algorithm}

\paragraph{Encryption.} The set of plaintexts is $\calP_{t,n}$, as defined in Section~\ref{subsec:notation}. Encryption is presented in Algorithm~\ref{algo:encrypt}. Notice that in steps 3-4, the computation of $\bfp'$ should be understood as a the generation of a uniform random vector such that $\rowspan_\bfg(\bfp')$ is the rowspan of $\bfP$.

\begin{algorithm}[t!]
  \KwData{public key $\kpub \in \FF_q^{n-k}$, plaintext $\bfP \in \calP_{t,n}$}
  \KwResult{ciphertext $\bfu \in \FF_q^{n-k}$}
  Compute any $\bfy \in \FF_q^n$ such that $\bfH \bfy^\top = \kpub^\top$\\
  Pick $\bfT \leftarrow_{\$} \{ \bfM \in \FF_2^{n \times n}, \deg_\bfg(\bfM) = \ell \}$\\
  Pick $\bfS \leftarrow_{\$} \{ \bfM \in \FF_2^{n \times n}, \rk(\bfM) = n \}$\\
  Compute $\bfp' = \bfg \bfS \bfP \in \FF_q^n$\\
  Output $\bfu \in \FF_q^{n-k}$ such that $\bfu^\top = \bfH (\bfy \bfT + \bfp')^\top$
  \caption{\label{algo:encrypt}$\Encrypt(\kpub, \bfP)$}
\end{algorithm}

\paragraph{Decryption.} We present in Algorithm~\ref{algo:decrypt} a decryption algorithm which may fail with negligible probability. The failure rate is devoted to be cryptographically small, and is bounded in Section~\ref{subsec:consistency}. We also make use of an $\FF_2$-linear map $V_\kpriv : \FF_q^n \to \FF_q^n$ such that $V_\kpriv(\kpriv) = \bfzero$. This map can be efficiently computed from the knowledge of the private key $\kpriv$. Mathematical properties of this map are given in Section~\ref{subsec:maths}

\begin{algorithm}[t!]
  \KwData{private key $\kpriv \in \FF_q^n$, ciphertext $\bfu \in \FF_q^{n-k}$}
  \KwResult{plaintext $\hat{\bfP} \in \calP_{t,n}$, or failure}
  Compute a solution $\bfx \in \FF_q^n$ to the linear system $\bfH \bfx^\top = \bfu^\top$.\\
  Compute $\bfz = V_\kpriv(\bfx) \in \FF_q^n$.\\
  Decode $\bfz$ as a corrupted $\Gab_{k+\ell+w}(\bfg)$-codeword. If success, one gets an error vector $\bfa \in \FF_q^n$ of rank $\le t$.\\
  {\bf If} $\rk(\bfa) < t$, output failure.\\
  {\bf Otherwise}, output  $\hat{\bfP} = {\tt RREF}(\ext_\bfg(\bfa))$.
  \caption{\label{algo:decrypt}$\Decrypt(\kpriv, \bfu)$}
\end{algorithm}

In Algorithm~\ref{algo:decrypt}, one needs to decode Gabidulin codes up to half their minimum distance, \emph{i.e.}\ to decode errors of rank less than $\lfloor \frac{n - \dim \Gab}{2} \rfloor$. Many such algorithms can be found in the literature since the seminal work of Gabidulin~\cite{Gabidulin85}. Some of them are based on solving a so-called \emph{key equation}, such as~\cite{Roth91, ParamonovT91, Gabidulin91, RichterP04, Wachter-ZehAS13} and others use interpolation, for instance~\cite{Loidreau05}. Fastest ones run in $\calO(n^2)$ operations over $\FF_q$.

  
\section{Parameters}
\label{sec:parameters}

In Table~\ref{tab:parameters}, we propose three sets of parameters for RAMESSES as a KEM, according to the desired level of security. Table~\ref{tab:parameters-safe} proposes a set of parameters for RAMESSES as a PKE. There are generic transformations from PKEs to KEMs, widely used in the NIST competition. Note that the decryption failure can be finely tuned as is explained in Section~\ref{subsec:consistency}. One can notice that the post-quantum security is much larger than half the classical one, which is unusual in code-based systems. Indeed the best current attacks against RAMESSES do not use enumeration techniques, which would benefit from the use of Grover algorithm, but Groebner bases algebraic techniques for which there is no known efficient quantum algorithmic speedup, as explained in Section~\ref{subsec:security}. 

\begin{table}[t!]
  \footnotesize
  \centering
  \begin{tabular}{ccccccccc}
    \hline
    \multirow{2}{*}{$n$} & \multirow{2}{*}{$k$} & \multirow{2}{*}{$w$} & \multirow{2}{*}{$\ell$} & \multirow{2}{*}{$t$} & classical  & post-quantum  & public key/ciphertext & private key  \\
    & & & & & security (bits) & security (bits) & size (bytes) & size (bytes)  \\
    \hline
    \hline
    $64$ & $32$ & $19$ & $3$ & $5$ & $141 \; (\ge 128)$ & $126$ & $256$ & $152$  \\
    $80$ & $40$ & $23$ & $3$ & $7$ & $202 \; (\ge 192)$ & $158$ & $400$ & $230$  \\
    $96$ & $48$ & $27$ & $3$ & $9$ & $265 \; (\ge 256)$ & $190$ & $576$ & $324$  \\
    \hline
  \end{tabular}
  \medskip
  
  \caption{\label{tab:parameters}Sets of parameters for RAMESSES as a KEM, with different levels of security. The security is estimated according to the current state of the art of algebraic attacks; the linear algebra constant is set to $\omega = \log_2(7) \simeq 2.807$. Decryption failure rates are respectively bounded by $2^{-40}$, $2^{-50}$ and $2^{-60}$. }
\end{table}

\begin{table}[t!]
  \footnotesize
  \centering
  \begin{tabular}{ccccccccc}
    \hline
    \multirow{2}{*}{$n$} & \multirow{2}{*}{$k$} & \multirow{2}{*}{$w$} & \multirow{2}{*}{$\ell$} & \multirow{2}{*}{$t$} & classical  & post-quantum  & public key/ciphertext & private key  \\
    & & & & & security (bits) & security (bits) & size (bytes) & size (bytes)  \\
    \hline
    \hline
    164 & 116 & 27 & 3 & 9 & $\ge$ 256 & $\ge$ 256 & 984 & 554 \\
    \hline
  \end{tabular}
  \medskip  
  \caption{\label{tab:parameters-safe} A set of parameters for RAMESSES as a PKE, with decryption failure rate $\le 2^{-128}$. The security is estimated according to the current state of the art of algebraic attacks; the linear algebra constant is set to $\omega = \log_2(7) \simeq 2.807$.}
\end{table}

\paragraph{Claimed security.} The claimed security is computed according to known attacks reported in Section~\ref{subsec:security}.

\paragraph{Public key size.} The public key consists in a vector $\kpub \in \FF_q^{n-k}$. Thus, its size is $(n-k)n$ bits, or $\frac{(n-k)n}{8}$ bytes.

\paragraph{Private key size.} For the private key $\kpriv \in \FF_q^n$, Alice actually needs to store only the map $V_\kpriv$. From Section~\ref{subsec:maths}, this map is a monic polynomial over $\FF_q$ of degree $w$. Hence only $w$ coefficients over $\FF_q$ actually need to be stored, the size of the private key is thus $wn$ bits, or $\frac{wn}{8}$ bytes.

\paragraph{Ciphertext size.} The ciphertext is a vector $\bfu \in \FF_q^{n-k}$, hence its size is $(n-k)n$ bits, \emph{i.e.}\ $\frac{(n-k)n}{8}$ bytes.

\section{Analysis}

\subsection{Mathematical background}
\label{subsec:maths}

Gabidulin codes can be interpreted in the context of skew polynomial rings. Recall that $\theta$ represents the Frobenius automorphism $x \mapsto x^2$. The \emph{skew polynomial ring} $\FF_q[X; \theta]$, originally studied by \O re~\cite{Ore33a,Ore33b}, is the ring of univariate polynomials defined by the non-commutative multiplicative rule
\[
X \cdot a = \theta(a) \cdot X, \quad\quad a \in \FF_q\,.
\]
In our context, skew polynomials are also called linearized polynomials. One can define the evaluation of a skew polynomial $P = \sum_{i=0}^d a_i X^i \in \FF_q[X; \theta]$ at $x \in \FF_q$ as follows:
\[
P(x) \mydef \sum_{i=0}^d a_i \theta^i(x) = \sum_{i=0}^d a_i x^{2^i}\,.
\]

The evaluation vector of $P$ at $\bfx \in \FF_q^n$ is defined as
\[
P(\bfx) \mydef (P(x_1), \dots, P(x_n)) \in \FF_q^n\,.
\]
Thus, the rows of ${\rm Moore}_n(\bfg)$ can be seen as the evaluation vectors over $\bfg$, of the sequence of degree-ordered skew monomials $1, X, \dots, X^{n-1}$. As a consequence, one can view Gabidulin codes as analogues of Reed-Solomon codes for skew polynomial rings:
\[
\Gab_k(\bfg) = \{ P(\bfg) \mid P \in \FF_q[X; \theta], \deg P < k \} \,.
\]

For $\bfx \in \FF_q^n$, the polynomial $P(X) \in \FF_q[X;\theta]$ of minimum degree such that $P(\bfg) = \bfx$ is the \emph{$\bfg$-interpolating polynomial of $\bfx$} and is denoted $L_\bfx(X)$. By definition $\deg(L_\bfx) = \deg_\bfg(\bfx)$.

Finally,  given $\bfe \in \FF_q^n$, the set of polynomials $P \in \FF_q[X; \theta]$ satisfying $P(\bfe) = \bfzero$ is a left-ideal $I_\bfe$ of $\FF_q[X; \theta]$. Since skew polynomial rings are principal ideal domains, we can define the \emph{minimum vanishing polynomial} $V_\bfe(X) \in \FF_q[X; \theta]$ of $\bfe$ as the unique monic skew polynomial which generates $I_\bfe$. Notice that $\deg(V_\bfe) = \rk(\bfe) \ge n - \deg_\bfg(\bfe)$.

The following lemma will be helpful for the analysis of the scheme consistency.

\begin{lemma}
  \label{lem:aux-consistency}
  Let $P(X) \in \FF_q[X; \theta]$ and $\bfa \in \FF_q^n$. Then we have  $\rowspan_\bfg( P(\bfa) ) \subseteq \rowspan_\bfg( \bfa )$. Moreover,  if $\rowspan_\bfg( P(\bfa) ) \ne \rowspan_\bfg( \bfa )$, then there exists a non-zero $x = \sum_{i=1}^n \lambda_i a_i \in \colspan( \bfa )$ such that $P(x) = 0$.
\end{lemma}
\begin{proof}
  Let $\bfB \in \FF_2^{n \times n}$ satisfy $\rowspan_\bfg(\bfa) = \{ \bfx \in \FF_2^n, \bfx \bfB = \bfzero \}$. In particular, one can see that $\bfa \bfB = \bfzero$. Hence, by $\FF_2$-linearity $P(\bfa) \bfB = P(\bfa \bfB) = \bfzero$. Thus, every $\bfy \in \rowspan_\bfg(P(\bfa))$ satisfies $\bfy \bfB = \bfzero$, leading to $\rowspan_\bfg(P(\bfa)) \subseteq \rowspan_\bfg(\bfa)$.

Assume now that $\rowspan_\bfg( P(\bfa) ) \ne \rowspan_\bfg( \bfa )$. It implies that $\dim \colspan(P(\bfa)) < \dim \colspan( \bfa )$. Let $(a_{i_j})_{1 \le j \le k} \subset \FF_q$ be an ordered basis of $\colspan( \bfa ) \subseteq \FF_q$ over $\FF_2$. Then there must exists a non-zero $(\lambda_j) \in \FF_2^k$ such that $\sum_{j=1}^k \lambda_j P(a_{i_j}) = 0$, otherwise we would have $\dim \colspan( P(\bfa) ) = k$. If we set $x = \sum_j \lambda_j a_{i_j} \in \FF_q \setminus \{ 0 \}$, then we get $P(x) = 0$ by $\FF_2$-linearity.\qed
\end{proof}

\subsection{Consistency}
\label{subsec:consistency}

In this section we characterize the output of algorithm $\Decrypt$ described in Section~\ref{sec:scheme}. As input, $\Decrypt$ receives a vector $\kpriv \in \FF_q^n$ of rank $w$ and a vector $\bfu \in \FF_q^{n-k}$ such that $\bfu = \bfH(\bfy \bfT + \bfp')^\top$, where
\begin{itemize}[label=--]
\item vector $\bfy \in \FF_q^n$ satisfies $\bfH \bfy^\top = \bfH\kpriv^\top$,
\item matrix $\bfT \in \FF_2^{n \times n}$ has $\bfg$-degree $\ell$,
\item vector $\bfp' = \bfg \bfS \bfP \in \FF_q^n$ has rank $t \mydef \lfloor \frac{n-k-\ell-w}{2} \rfloor$.
\end{itemize}

First, notice that $\bfy = \kpriv + \bfc$ for some $\bfc \in \Gab_k(\bfg)$. In the first step of Algorithm~\ref{algo:decrypt}, a vector $\bfx \in \FF_q^n$ solution to $\bfH \bfx^\top = \bfu^\top$ is computed. One can see that the set $S$ of such solutions is
\[
S = \{ \bfy\bfT + \bfp' + \bfc' \mid \bfc' \in \Gab_k(\bfg) \} \subseteq \FF_q^n\,.
\]
Therefore, in step $2$ of Algorithm~\ref{algo:decrypt}, we have
\[
\begin{aligned}
  \bfz = V_\kpriv(\bfx) &=  V_\kpriv((\bfc+\kpriv)\bfT + \bfp' + \bfc')\\
  &= V_\kpriv(\bfc' + \bfc\bfT) + \underbrace{V_\kpriv(\kpriv)}_{\bfzero}\bfT + V_\kpriv(\bfp')\,.
  \end{aligned}
\]
We notably used the $\FF_2$-linearity of $V_\kpriv$. Also recall that, for any $\bfa \in \FF_q^n$, $L_\bfa(X)$ denotes the $\bfg$-interpolating polynomial of $\bfa$. Then we get:
\[
\bfz = (V_\kpriv \cdot (L_{\bfc'} + L_{\bfc \bfT}))(\bfg) + V_\kpriv(\bfp')\,.
\]
Moreover, $L_{\bfc \bfT} = L_\bfc \cdot L_{\bfg\bfT}$ yields $\deg(L_{\bfc \bfT}) \le k - 1 + \ell$ since $\deg_\bfg(\bfT) = \ell$. Therefore, the polynomial $V_\kpriv \cdot (L_{\bfc'} + L_{\bfc \bfT})$ has degree at most $\deg(V_\kpriv) + \max\{ \deg(L_{\bfc'}),  \deg(L_{\bfc \bfT}) \} \le w + k - 1 + \ell$.

We also know that $\rk(V_\kpriv(\bfp')) \le \rk(\bfp') = \rk(\bfP) = t = \lfloor \frac{n-k-\ell-w}{2} \rfloor$. Hence, in third step of Algorithm~\ref{algo:decrypt}, any decoding algorithm for $\Gab_{k+w+\ell}(\bfg)$ that decodes errors of rank at most $t$ will retrieve $V_\kpriv(\bfp')$ from $\bfz$. Finally, Algorithm~\ref{algo:decrypt} outputs a matrix $\hat{\bfP} \in \calP_{t,n}$ such that $\rowspan(\hat{\bfP}) = \rowspan_\bfg(V_\kpriv(\bfp'))$.

As a consequence, decryption fails whenever $\rowspan_\bfg(V_\kpriv(\bfp')) \ne \rowspan(\bfP)$, where $\bfP$ is the original plaintext. First notice that $\rowspan(\bfP) = \rowspan_\bfg(\bfp')$. Then, Lemma~\ref{lem:aux-consistency} shows that if decryption fails, then there exists a non-zero $x \in \colspan(\bfp')$ such that $V_\kpriv(x) = 0$. Let us now recall that the set of zeroes of $V_\kpriv$ is exactly $\colspan(\kpriv)$. Hence we get the following result.

\begin{lemma}
  \label{lem:equiv-failure}
  Let $\bfP \in \calP_{t,n}$. If, on input $(\kpriv, \Encrypt(\kpub, \bfP))$ where $(\kpub, \kpriv) \leftarrow \KeyGen$, algorithm $\Decrypt$ does not output $\bfP$, then  matrix $\bfS$ has been chosen at step $4$, such that $\colspan(\kpriv) \cap \colspan(\bfS \bfP)) \ne \{ 0 \}$.
\end{lemma}

One can now estimate the probability of failure of $\Decrypt$.

\begin{lemma}
  Let $(\kpub, \kpriv) \leftarrow \KeyGen$ be any pair of keys generated by $\KeyGen$, on public parameters $n,w,t$. Then, for every $\bfP \in \calP_{t,n}$,
  \[
  \PP_{\bfS, \bfT, \bfy} \left( \hat{\bfP} \ne \bfP \;\middle|\;
  \begin{array}{l}
    \bfu \leftarrow \Encrypt(\kpub, \bfP)\\
    \hat{\bfP} \leftarrow \Decrypt(\kpriv, \bfu)
  \end{array}
  \right) \le 2^{-(n-t-w)}\,.
  \]
\end{lemma}

\begin{proof}
  Using Lemma~\ref{lem:equiv-failure}, we have
  \[
  \begin{aligned}
   \PP_{\bfS, \bfT, \bfy} &\left( \hat{\bfP} \ne \bfP \;\middle|\;
   \begin{array}{l}
     \bfu \leftarrow \Encrypt(\kpub, \bfP)\\
     \hat{\bfP} \leftarrow \Decrypt(\kpriv, \bfu)
   \end{array}
   \right)
   ~\\ &=  \PP_\bfS\big( \colspan( \kpriv ) \cap \colspan( \bfS \bfP ) \ne \{ 0 \} \big)\,.
  \end{aligned}
  \]
  It is easy to check that the probability that a $t$-dimensional random subspace of $\FF_2^n$ intersects non-trivially a fixed subspace of dimension $w$ is bounded by $\frac{(2^t-1)(2^w-1)}{2^n-1} \le 2^{t+w-n}$. This concludes the proof.\qed
\end{proof}



\subsection{Security proof}
\label{subsec:securityproof}

Let us first introduce two problems to which the security of RAMESSES can be reduced. Problem~\ref{prob:corgab} is an ad hoc problem. The search version of Problem~\ref{Prob:DecodeGab} corresponds to decoding errors of rank $w$ in a Gabidulin code; this problem is believed hard for $w$ between $\frac{n-k}{2}$ and $n-k$, and a improvement in solving this problem would be significant in coding theory.
\begin{problem}[Syndrome correlation for Gabidulin codes ({\corgab})]\\
  \label{prob:corgab}  
  Let $\bfH \in \FF_q^{(n-k) \times n}$ be a fixed parity-check matrix of $\Gab_k(\bfg)$, and $1 \le \ell \le n-1$.
  \begin{itemize}
  \item {\bf Input:} access to distributions
    \begin{enumerate}
    \item $\mathcal{D}_1$: $(\bfH \bfx^\top, \bfH \bfT^\top \bfx^\top)$, where $\bfx \leftarrow_\$ \FF_q^n$ and $\bfT  \leftarrow_\$ \calM_\ell$,
    \item $\mathcal{D}_2$: $(\bfH \bfx^\top, \bfr^\top)$, where $\bfx \leftarrow_\$ \FF_q^n$ and $\bfr \leftarrow_\$ \FF_q^{n-k}$.
    \end{enumerate}
  \item 
    {\bf Goal:} distinguish between $\mathcal{D}_1$ and $\mathcal{D}_2$.
  \end{itemize}
\end{problem}

\begin{problem}[Syndrome decoding for Gabidulin codes, {\gabsd}]\\
  \label{Prob:DecodeGab}
   Let $\bfH \in \FF_q^{(n-k) \times n}$ be a fixed parity-check matrix of $\Gab_k(\bfg)$, and $\frac{n-k}{2} < w < n-k$.
\begin{itemize}
\item  {\bf Input:} access to distributions
    \begin{enumerate}
    \item $\mathcal{D}_1$: $\bfH \bfx^\top$, where $\bfx \leftarrow_\$ S_w$,
    \item $\mathcal{D}_2$: $\bfr^\top$, where $\bfr \leftarrow_\$ \FF_q^{n-k}$.
    \end{enumerate}
\item { \bf Goal:} distinguish between $\mathcal{D}_1$ and $\mathcal{D}_2$.
\end{itemize}
\end{problem}

We now shortly show the indistinguishability under chosen plaintext attacks (IND-CPA) of RAMESSES with the following sequence of games.
\begin{itemize}
\item[] {\bf Game 0.} The real scheme with plaintext $\bfP$.
\item[] {\bf Game 1.} We modify {\bf Game 0} as follows. In the key generation, the vector $\kpriv$ is now picked uniformly at random in $\FF_q^n$, without any rank constraint.  
 \item[] {\bf Game 2.} We modify {\bf Game 1} as follows. In the encryption algorithm, $\bfH (\bfy\bfT + \bfg\bfS\bfP_{(1)})^\top$ is replaced by $\bfr + \bfH(\bfg\bfS\bfP_{(1)})^\top$, where $\bfr$ is generated uniformly at random in $\FF_q^{n-k}$.
 \item[] {\bf Game 3.} We modify {\bf Game 2} as follows. The plaintext $\bfP_{(1)}$ is replaced by the plaintext $\bfP_{(2)}$.
 \item[] {\bf Game 4.} This game is identical to {\bf Game 1}, except that the plaintext is $\bfP_{(1)}$ is replaced by the plaintext $\bfP_{(2)}$.
 \item[] {\bf Game 5.} The real scheme with plaintext $\bfP_{(2)}$.
\end{itemize}
One can then prove that the advantage ${\rm Adv}_\calA^{\bf Dist}$ for an adversary $\calA$ to distinguishing the encryption of $\bfP_{(1)}$ and $\bfP_{(2)}$ satisfies:
\[
{\rm Adv}_\calA^{\bf Dist} \le 2 ( {\rm Adv}_\calA^{\gabsd} + {\rm Adv}_\calA^{\corgab})\,.
\]
Roughly speaking, one actually mimics the security proof given in~\cite{AguilarBDGZ18}. The $2 {\rm Adv}_\calA^{\gabsd}$ term comes from transitions between games 0 and 1, and games 4 and 5, whereas transitions between games 1 and 2, and games 3 and 4 yield the $2 {\rm Adv}_\calA^{\corgab}$ term. Games 2 and 3 are information-theoretically indistinguishable since $\bfr$ is random.

\subsection{Existing attacks}
\label{subsec:security}


In the following, we denote by $\lambda$ the desired security parameter, \emph{i.e.}, any attack against the cryptosystem must cost at least $2^\lambda$ operations over $\FF_2$.

\paragraph{Exhaustive search attacks.} In order to avoid attacks by exhaustive search, one has the following constraints on the parameters.
\begin{enumerate}
\item $| \calP_{t,n}| = \qbin{n}{t}{2} \ge 2^\lambda$, satisfied when $t(n-t) \ge \lambda$.
\item $|\{ \kpriv \}| \ge \qbin{n}{w}{2} \ge 2^\lambda$, satisfied when $w(n-w) \ge \lambda$.
\item $|\calM_\ell| \ge 2^\lambda$, satisfied when $(\ell+1)n \ge \lambda $.
\end{enumerate}

\paragraph{Attack by decoding beyond the unique decoding radius of Gabidulin codes.} Let $\bfe' \in \FF_q^n$ be \emph{any} solution of $\bfH \bfe'^\top = \kpub^\top$ of rank $\le w$. From the consistency analysis one can see that $\bfe'$ can be used as an alternate private key in the $\Decrypt$ algorithm. The computation of such a vector $\bfe'$ actually corresponds to the search version of {\gabsd} problem.

This problem is easy for $w \le \lfloor \frac{n-k}{2} \rfloor$ (it corresponds to half-minimum-distance decoding) and for $w \ge n-k$ (equivalent to interpolation for linearized polynomials). For our concern, we have $\lfloor \frac{n-k}{2} \rfloor < w < n-k$, and we believe that the search version of {\gabsd} is hard in this range of parameters.

A solution consists in enumerating vector spaces of dimension slightly higher than $w$, checking whether they guessed correctly a large part of the solution space, and in such case, interpolating the solution. Roughly speaking, the number of valid choices for the subspace is large, but the complexity of finding one remains exponential in the code length. Precisely, in our settings ($m=n$, and $n-k$ even) the number of vector spaces to test before finding one solution is on average
\[
   \mathcal{N}_{Class-{\gabsd}} \approx 0.3\cdot 2^{\delta(n+k-2\delta)}\,,
\]
where $\delta \mydef w - \lfloor \frac{n-k}{2} \rfloor > 0$. This quantity is used as a bound for the complexity of solving {\gabsd}. By using a straightforward Grover algorithm, we obtain that the number of iterations to be completed on a quantum computer is roughly 
\[
   \mathcal{N}_{Quant-{\gabsd}} \approx 0.55\cdot 2^{\frac{\delta}{2}(n+k-2\delta)}\,.
\]

\paragraph{Attack \emph{via} a reduction to a quadratic system over $\FF_2$.} Given a vector $\bfe \in \FF_q^n$ with  $\rk(\bfe) = w$, any solution $\bfy \in \FF_q^n$ to $\bfH \bfy^\top = \bfH \bfe^\top$ can be written as $\bfy = \bfc + \bfe$ for some $\bfc \in \Gab_{k+\ell}(\bfg)$. Therefore, $\bfy$ satisfies
\begin{equation}
  \label{eq:key-eq}
V_\bfe(\bfy) = (V_\bfe \cdot U)(\bfg)\,, 
\end{equation}
where $U(X) = \sum_{i=0}^{k+\ell-1} u_i X^i \in \FF_q[X, \theta]$. Hence, an attack would consist in searching for $V_\bfe$ and $U$ in the previous equation, for some fixed $\bfy$ solution to $\bfH \bfy^\top = \kpub^\top$.

Equation~\eqref{eq:key-eq} can be turned into a quadratic system over $\FF_2$ (see Appendix~\ref{app:key-recovery-by-quadratic-system} for details). Using results of Bardet \emph{et al.}~\cite{BardetFSS13}, the solving complexity would be in $\calO(2^{0.561\, n^2})$, which remains much larger than the complexity of the previous attack.\footnote{However, notice that the system to be solved in~\cite{BardetFSS13} is assumed random, and such that no specialization of variables can be made. This is unlikely the case for our system, but it requires a finer analysis --- which is not the scope of this paper --- to understand whether improvements can be made in order to solve the system.}

\paragraph{Attack \emph{via} a reduction to a {\minrank} instance.} The recovery of a representative $\bfp' = \bfg \bfS \bfP \in \FF_q^n$ of the plaintext $\bfP$, given only a ciphertext $\bfu$ and $\kpriv$, can be modeled as follows. First, one computes (i) any solution $\bfx \in \FF_q^n$ of $\bfH \bfx^\top = \bfu^\top$, and (ii) any solution $\bfy \in \FF_q^n$ to $\bfH\bfy^\top = \kpub^\top$. Due to the form of the ciphertext, this leads us to
\begin{equation}
  \label{eq:base-minrank}
\bfx - \bfy\bfT - \bfc  = \bfp' \,,
\end{equation}
where $\bfc \in \Gab_{k+\ell}(\bfg)$ and $\bfT \in \FF_2^{n \times n}$ are unknown to the attacker. Notice that $\bfT$ lies in a $\FF_2$-vector space of dimension $(\ell+1)n$, since $\bfg\bfT \in \Gab_{\ell+1}(\bfg)$. Two kinds of attacks can then be mounted to solve~\eqref{eq:base-minrank}.
  
First, Equation~\eqref{eq:base-minrank} can be written $\bfx = (\bfc + \bfy \bfT) + \bfp'$, which means that the problem can be rephrased as decoding an error $\bfp'$ of rank $t$ in the underlying code
\[
\calD \mydef \Gab_{k+\ell}(\bfg) + {\rm span}_{\FF_2}\big( \{ \bfy \bfT \mid \bfT \in \calM_\ell \} \big)\,.
\]
Notice that $\calD \subseteq \FF_q^n$ is an $\FF_2$-linear code of $\FF_2$-dimension at most $(k + 2\ell + 1)n$. One can then write $\bfy \bfT = L_\bfy(\bfg \bfT)$, which yields $\calD = \Gab_{k+\ell}(\bfg) + L_\bfy(\Gab_{\ell+1}(\bfg))$.  A straightforward decoding approach would lead to an attack in time roughly $2^{kr}$. One could also try to decode in the smallest $\FF_q$-linear code containing $\calD$, and use the additional structure provided by the $\FF_q$-linearity. This structure has been widely employed in the recent improvements, see~\cite{BardetBBGNRT09}. However, it is unlikely that the $\FF_q$-dimension of ${\rm span}_{\FF_q}(\calD) = \Gab_{k+\ell}(\bfg) + {\rm span}_{\FF_q}(L_\bfy(\Gab_{\ell+1}(\bfg)))$ is small, since the $\FF_2$-endomorphism of $\FF_q[X,\theta]$ defined by $P \mapsto L_\bfy P$ is not $\FF_q$-linear.

Second, one can see Equation~\eqref{eq:base-minrank} as an instance of {\minrank}, a problem formally introduced by Courtois in~\cite{Courtois01} after the cryptanalysis of HFE~\cite{KipnisS99}.
\begin{problem}[{\minrank} search problem] Let $\KK$ be a field.
\label{prob:minrank}
  \begin{itemize}
   \item  {\bf Input:} $\bfM_0, \bfM_1, \dots, \bfM_K \in \KK^{N \times n}$ and an integer $t$.
   \item  {\bf Goal:} Find $(x_1, \dots, x_K) \in \KK^K$ such that $\rk_\KK(\bfM_0 - \sum_{i=1}^K x_i \bfM_i) \le t$.
  \end{itemize}
\end{problem}
 Let us denote by $\{ \bfT_1, \dots, \bfT_{n(\ell+1)} \} \subseteq \FF_2^{n\times n}$ an $\FF_2$-basis of $\Gab_{\ell+1}(\bfg)$, the smallest vector space containing $\calM_\ell$. Similarly, $\ext_\bfg(\bfc)$ can be written in some basis $\{ \bfC_1, \dots, \bfC_{n(k+\ell)} \} \subseteq \FF_2^{n\times n}$ of the $\FF_2$-vector space of dimension $n(k+\ell)$ representing $\Gab_{k+\ell}(\bfg)$. Applying $\ext_\bfg$ to Equation~\eqref{eq:base-minrank}, we get:
\[
\bfX - \sum_{i=1}^{n(\ell+1)} t_i \bfY \bfT_i -  \sum_{i=1}^{n(k+\ell)} c_i \bfC_i  = \bfP',
\]
where $(\bfX, \bfY, \bfP') = (\ext_\bfg(\bfx), \ext_\bfg(\bfy), \ext_\bfg(\bfp'))$. Since $\rk(\bfP') =t$, one gets an instance of the {\minrank} problem, with one \enquote{base matrix} $\bfX \in \FF_2^{n \times n}$ and $K \mydef n(k+2\ell+1)$ \enquote{summand matrices} $\{ \bfY\bfT_1, \dots, \bfY\bfT_{n(\ell+1)}, \bfC_1, \dots, \bfC_{n(k+\ell)} \}$.

There exist several approaches to solve the {\minrank} problem.  In~\cite{GoubinC00}, Goubin and Courtois gave an algorithm which finds a solution in expected time $\calO(K^32^{t \lceil K/n\rceil})$. In 1999, Kipnis and Shamir~\cite{KipnisS99} proposed a multivariate formulation of {\minrank} which can be solved by computing Groebner bases. Such computations can be run in time $\calO(\binom{m+d-1}{d}^\omega)$, where $2 \le \omega < 3$ is the linear algebra constant, $m = t(n-t)+K$ and  $d$ is the \emph{degree of regularity} of the system~\cite{Lazard83}. Faugère, Levy-dit-Vehel and Perret~\cite{FaugereLP08} proved that, in the Kapnis-Shamir formalism, any instance can be reduced to a simpler one if $\Delta \mydef K - (n-t)^2 > 0$. In our case, setting $w \ge \ell + 1$ ensures that $\Delta \le 0$. Moreover, the authors proved that the degree of regularity is lower than what is expected for random systems, and it seems to be upper bounded by $t+2$ heuristically. This heuristic was confirmed by Verbel \emph{et al.}~\cite{VerbelBCPS19} for superdetermined instances, and by Bardet \emph{et al.}~\cite{BardetBBGNRT09} in the context of decoding low rank errors in random codes. Finally, the latter work also presents instances for which the solving degree decreases to $d = t$. We choose to consider this conservative setting; the running time for the computation of the associated Groebner basis is thus in
\[
\calO\left(\;\binom{t(n-t) + n(k+2\ell+1) +t-1}{t}^\omega\;\right)\,.
\]

To sum up, the reduction to {\minrank} leads us to the following bounds on the parameters:
\[
w \ge \ell + 1,\quad\quad   \omega \cdot \log \tbinom{n(k+2\ell+t+1)-t^2+t-1}{t}  \ge \lambda ,\quad\quad t(k+2\ell+1)  \ge \lambda\,.
\]

\section{Conclusion}

The parameters we proposed for RAMESSES are deliberately aggressive so that to encourage research in studying the security of the encryption scheme. The simplicity and versatility of the scheme enables very efficient tuning for many sets of parameters, without making them grow prohibitively. Namely, the size of keys and ciphertexts grow linearly with the security parameter; this is usually not the case in other code-based encryption schemes where sizes (notably the public key size) grow quadratically with the security level, except for systems using structural tricks to reduce the key size. 



\bibliographystyle{abbrv}
\bibliography{biblio}

\appendix

\section{A quadratic system model}

\label{app:key-recovery-by-quadratic-system}

Without loss of generality, we here assume that $\bfg = (g^{q^0}, g^{q^1}, \dots, g^{q^{n-1}})$ is a normal basis of $\FF_q/\FF_2$. Recall that any solution $\bfy \in \FF_q^n$ to $\bfH \bfy^\top = \bfH \bfe^\top$ satisfies
\begin{equation}
  \label{eq:key-eq-2}
V_\bfe(\bfy) = (V_\bfe \cdot U)(\bfg)\,, 
\end{equation}
where $\rk(\bfe) = w$, and $U(X) = \sum_{i=0}^{k+\ell-1} u_i X^i \in \FF_q[X, \theta]$. The attack would consist in searching for $V_\bfe$ and $U$ in the previous equation, for some fixed $\bfy$ solution to $\bfH \bfy^\top = \kpub^\top$.

Let us now write $V_\bfe(X) = \sum_{i=0}^w v_i X^i$, and denote by $(v_{r,i})_{1 \le r \le n} \in \FF_2^n$ (resp. $(u_{s,j})_{1\le s \le n}$) the decomposition of $v_i$ (resp. $u_j$) in the basis $\bfg$. Then, Equation~\eqref{eq:key-eq-2} rewrites
\begin{equation}
  \label{eq:quadratic}
\Big( \sum_{i=0}^w \sum_{r=1}^n v_{r,i} \bfB^{q^r} \bfC^i \Big) \ext_\bfg(\bfy) = \sum_{i=0}^w \sum_{j=0}^{k+\ell-1} \sum_{r=1}^n \sum_{s=1}^n v_{r,i} u_{s,j} \;\bfB^{q^r+q^{s+i}} \bfC^{i+\ell}\,,
\end{equation}
where $\bfB = \ext_\bfg(g \cdot \bfg) \in \FF_2^{n\times n}$ is the matrix of the multiplication by $g$ in $\FF_q$, and $\bfC \in \FF_2^{n\times n}$ is the right-cyclic-shift matrix. Equation~\eqref{eq:quadratic} thus defines a quadratic system over $\FF_2$, with $n(w+k+\ell+1)$ unknowns coefficients $(v_{r,i})_{r,i}$ and $(u_{s,j})_{s,j}$, involved in $n^2$ equations (the coefficients of the matrices).

Though the above system is not random (random systems are believed to be the hardest), we report one result concerning the complexity of random Boolean quadratic systems. In~\cite{BardetFSS13}, Bardet \emph{et al.} gave an algorithm solving such a system. Without any specialization of variables, its running time is in $O(2^{2H(M(\alpha))N_{\rm var}})$, where $N_{\rm var}$ is the number of variables, $\alpha \mydef  N_{\rm eq} / N_{\rm var}$ is the ratio between equations and variables,
\[
H(t) \mydef - t \log_2(t) - (1-t)\log_2(1-t)
\]
is the binary entropy function, and
\[
M(x) \mydef -x + \frac{1}{2} + \frac{1}{2}\sqrt{2x^2-10x-1+2(x+2)\sqrt{x(x+2)}}\,.
\]

In our case, $N_{\rm var} = \frac{n^2}{\alpha}$ and the parameters have been chosen such that $\alpha = \frac{n}{k+w+\ell+1} \in (1.0,1.33)$. It leads to $H(M(\alpha)) > 0.372$, hence $\frac{2H(M(\alpha))}{\alpha} > 0.561$. In other terms, this approach leads to:
\[
0.561\, n^2 \ge \lambda
\]
under the assumptions that the system \emph{behaves like} a random system and we do not specialize variables.

\end{document}